\documentclass[letterpaper, 10 pt, conference]{ieeeconf}  % Comment this line out

\IEEEoverridecommandlockouts                        
\usepackage{amsmath}
\usepackage{amssymb}
\usepackage{amsfonts}
\usepackage{graphicx}
\usepackage{algorithm}
\usepackage{caption}
\usepackage{subcaption}
\usepackage[noend]{algpseudocode}
\usepackage{cite}
\usepackage{xcolor}
\usepackage[short]{optidef}

\newcommand{\norm}[1]{\left\lVert#1\right\rVert}

\newtheorem{proposition}{Proposition}

\newtheorem{assumption}{Assumption}

\title{\LARGE \bf
A Semidefinite Programming Approach to Discrete-time\\ Infinite Horizon Persistent Monitoring}

\author{Samuel C. Pinto$^1$, Sean B. Andersson$^{1,2}$, Julien M. Hendrickx$^3$, and Christos G. Cassandras$^{2,4}$% <-this % stops a space
\\
$^1$Dept. of Mechanical Engineering, $^2$Division of Systems Engineering,\\ $^4$Dept. of Electrical and Computer Engineering \\
Boston University, Boston, MA 02215, USA \\
$^3$ICTEAM Institute, UCLouvain, Louvain-la-Neuve 1348, Belgium \\
\{samcerq,sanderss,cgc\}@bu.edu, julien.hendrickx@uclouvain.be
\thanks{This work was supported in part by NSF under grants ECCS-1931600, DMS-1664644, CNS-1645681, and CMMI-1562031, by ARPA-E under grant DE-AR0001282, by AFOSR under grant FA9550-19-1-0158,  and by the MathWorks. The work of J. Hendrickx was supported by the “RevealFlight” Concerted Research Action (ARC) of the Federation Wallonie-Bruxelles, by the Incentive Grant for Scientific Research (MIS) “Learning from Pairwise Comparisons” of the F.R.S.-FNRS.
}
}

\begin{document}

\maketitle
\thispagestyle{empty}
\pagestyle{empty}

%%%%%%%%%%%%%%%%%%%%%%%%%%%%%%%%%%%%%%%%%%%%%%%%%%%%%%%%%%%%%%%%%%%%%%%%%%%%%%%%
\begin{abstract}
We investigate the problem of persistent monitoring, where a mobile agent has to survey multiple targets in an environment in order to estimate their internal states. These internal states evolve with linear stochastic dynamics and the agent can observe them with a linear observation model. However, the signal to noise ratio is a monotonically decreasing function of the distance between the agent and the target. The goal is to minimize the uncertainty in the state estimates over the infinite horizon. We show that, for a periodic trajectory with fixed cycle length, the problem can be formulated as a set of semidefinite programs. We design a scheme that leverages the spatial configuration of the targets to guide the search over this set of optimization problems in order to provide efficient trajectories. Results are compared to a state of the art approach and we obtain improvements of up to 91\% in terms of cost in a simple scenario, with much lower computational time.
\end{abstract}

%%%%%%%%%%%%%%%%%%%%%%%%%%%%%%%%%%%%%%%%%%%%%%%%%%%%%%%%%%%%%%%%%%%%%%%%%%%%%%%%
\section{INTRODUCTION}
\label{sec:intro}
We study the problem of persistent monitoring of a finite set of targets by a mobile agent, where each target has an internal state that evolves over time with some degree of uncertainty in its dynamics. The agent can observe these internal states when it is close to a given target, however, in order to monitor all the targets, it needs to move around the environment and visit each one of them infinitely often. This paradigm finds applications in a wide range of domains, such as observing multiple nanoparticles using a confocal microscope \cite{shen2010tracking}, ocean temperature monitoring \cite{lan2016rapidly}, urban surveillance using drones \cite{kim2018designing}, and pipeline inspection \cite{ostertag2019robust}.

This problem is closely related to the Traveling Salesman Problem (TSP) \cite{applegate2006traveling}, where, given a set of targets (possibly constrained to a graph-based structure), the goal is to find a cycle in which the agents efficiently visit all the targets in order to minimize the traveled distance or total travel time. The major difference between the TSP and the problem we are dealing with in this paper is that the optimization goal we consider is to minimize the uncertainty rather than distance. We emphasize that due to the dynamic nature of uncertainty, we cannot model the agent as being embedded in a graph with fixed cost on the edges, as assumed in TSP.% The present work has also strong connections to the sensor allocation problem \cite{le2010scheduling}, where sensors can observe the targets, but there are not enough sensors to continously monitor all the targets. The sensor allocation problem, however, assumes that sensors are located at fixed positions, and therefore does not incorporate the effect of the agent movement (i.e. mobile sensor) in the formulation. 

In the realm of persistent monitoring, significant previous work has been done. In \cite{lan2016rapidly}, a variant of the Rapidly Exploring Random Trees (RRT) algorithm denoted Rapid Random Cycles (RRC) was designed for cyclic discrete time persistent monitoring while \cite{lan2014variational} proposed an optimal control approach for the continuous time version of the persistent monitoring problem that relied on a solution of the two-point boundary value problem resulting from a Hamiltonian analysis. However, the solution of the two-point boundary value problem is numerically challenging and computationally expensive. %Also, a recent work \cite{ostertag2019robust} has introduced a technique to optimize steady state uncertainty over a pre-defined path, by adjusting velocities using a local optimization algorithm. This work, however, does not discuss how to generate the path.

Our group has also approached this problem from different point of views. Some of the initial work considered a simple uncertainty metric, where the target uncertainty grows or decreases linearly \cite{cassandras2013optimal}. More recent work has studied this problem from the assumption that the internal state of the target evolves according to linear, stochastic dynamics with linear observations, but corrupted with noise \cite{pinto2020multidimensional}. However, these works approached the problem from a continuous-time perspective. The trajectory of the agents was parameterized using a finite number of parameters, and then gradient descent was used to optimize over these parameters, leading to scalable solutions. The trajectories often converged to local optima and provided no insight on how far from global optimality the solutions were. To best of the authors knowledge, the only work that includes a notion of global (asymptotic probabilistic) optimality in persistent monitoring of uncertain states is \cite{lan2016rapidly} and traditional tools with global guarantees, such as dynamic programming, are computationally prohibitive even for a small number of targets. In this work, we intend to have more than a local notion of optimality while retaining the ability to generate efficient trajectories for problems with a small to moderate number of targets within a reasonable computational time. 

%In this work, we explore periodic trajectories, as in \cite{pinto2020multidimensional}, and approach the problem from an infinite horizon perspective. By optimizing only one period of a periodic trajectory and the respective steady state uncertainty, we can design efficient long term monitoring performance. 
In this paper, however, we limit ourselves to a single agent and a discrete time model and formulate the problem in a way that both the computation of the steady state uncertainty and the local optimization of the trajectory can be framed as a single optimization problem, a semidefinite program (SDP). The formulation and solution of this joint optimization problem is the main contribution of this paper. This contrasts with our previous approach \cite{pinto2020multidimensional} where, in a gradient descent scheme, the steady state uncertainty had to be computed through a computationally costly algorithm.
In the present work we benefit from efficient and reliable SDP solvers and are able to quickly solve a local version of the persistent monitoring problem.
Moreover, we are not limited to local optimality, as we have also embedded this local SDP-based optimizer into a higher level algorithm that searches globally for different periodic trajectories. This higher level scheme leverages the spatial distribution of the targets and then feeds the lower level optimization with configurations that will lead to feasible schedules. Due to the infinite number of candidate trajectories, we still are not able to guarantee global optimality. However, simulation results show that the approach proposed in this paper is able to efficiently handle problems with a small to moderate number of targets, providing trajectories with good performance even in the initial iterations of the higher level algorithm and also significantly improving them as it runs longer. Moreover, trajectories generated with the approach here proposed give a significant reduction (91\%) in terms of estimation error when compared to RRC \cite{lan2016rapidly}, while also showing significant computational time reduction. 

%The rest of this work is organized as follow: in Sec. \ref{sec:formulation} the problem formulation is presented. Sec. \ref{sec:ricatti_eq_opt} discusses how the computation of the estimation error can be viewed as an optimization problem. The joint optimization of the trajectory and the estimation error is discussed in Sec. \ref{sec:entire_opt}, along with a scheme to search for candidate cycles. In Sec. \ref{sec:results} simulation results are provided and our approach is compared with RRC. Finally, Sec. \ref{sec:conclusion} shares conclusions and ideas for future works.
 %%%%%%%%%%%%%%%%%%%%%%%%%%%%%%%%%%%%%%%%%%%%%%%%%%%%%%%%%%%%%%%%%%%%%%%%%%%%%%%%
\section{Problem Formulation}
\label{sec:formulation}

Consider an environment with $N$ target locations that should be monitored. Each target $i$ has an internal state $\phi_i\in\mathbb{R}^{L_i}$ that continuously evolves over time.
An agent can move and sense these targets. We assume that the dynamics of the targets and observations of its internal state by the agent are given according to the following discrete-time model:
\begin{subequations}
\label{eq:model}
\begin{equation}
    \phi_i(k+1)=A_i\phi_i(k) + w_i(k),
\end{equation}
\begin{equation}
    z_i(k)=H_i(k) \phi_i(k) + v_i(k),
\end{equation}
\end{subequations}
where $w_i(k)$ and $v_{i}(k)$ are zero mean, mutually independent white Gaussian processes with constant covariance matrices $Q_i$ and $R_i$, respectively.

Assume that at a given time step $k$, the agent has a position $s(k)$ and that the agent can deterministically control its position according to the following model
\begin{equation}
    \label{eq:agent_kinematics}
    s(k+1) = s(k) + u_k,\  \norm{u_k} \leq u_{\max}.
\end{equation}

We note that the choice of these simple dynamics is made for ease of presentation. Extending to more complicated agent dynamics is straightforward. Motivated by the fact that real sensors (such as sonar, cameras and lidars) normally have a finite range and that their observation quality usually decays as the agent moves farther from the sensing target, we assume the following model for $H_i(k)$:
\begin{equation}
\label{eq:sensing_model}
\begin{aligned}
    H_i(k) &= \sqrt{\gamma_i(k)}H_{i,\max},\\
    \gamma_i(k) &= \begin{cases}
    \left(1-\frac{(s(k)-x_i)^2}{r_i^2}\right),&\ \norm{s(k)-x_i} \leq r_i,\\
    0,&\ \text{otherwise}.
    \end{cases}
\end{aligned}
\end{equation}
In \eqref{eq:sensing_model}, $x_i$ is the position of target $i$ and $H_{i,\max}$ is a constant matrix.

This particular structure of $H_i(k)$ captures the fact that the power of the signal decays as the agent moves farther from the target, while the noise power stays constant. If the distance between agent and target is larger than $r_i$, the intensity of the signal is zero, which is equivalent to not sensing at all.
The particular quadratic decay was chosen due to the fact that it can be easily incorporated into an SDP. We plan to study extensions to more general decay shapes in future work.

%\subsection{Estimating the states}
When estimating the states, the goal is to find an unbiased estimator $\hat{\phi}_i(k)$ that minimizes the mean squared estimation error over an infinite time horizon. Letting ${\Sigma}_i(k)$ denote the covariance matrix of the estimator $\hat{\phi}_i(k)$ and noting that $E[(\hat{\phi}_i-\phi_i)^T(\hat{\phi}_i-\phi_i)]=\text{tr(}\Sigma_i\text{)}$, we define the cost $C$ of a particular infinite horizon trajectory of the agent as:
\begin{equation}
    \label{eq:cost}
    C= \lim_{M \rightarrow \infty} \frac{1}{M}\sum_{j=1}^M\sum_{i=1}^N\text{tr}(\Sigma_i(j))
\end{equation}
when the above limit exists. When it does not exist, the cost is defined as infinity.

Since both the dynamics and the observations are linear and the process and observation noises are Gaussian and uncorrelated, with the additional assumption that the distribution of $\phi_i(0)$ is Gaussian and uncorrelated with $w_i(k)$ and $v_i(k)$, we know that the optimal estimator is a Kalman Filter \cite{thrun2002probabilistic}. %The Kalman filter has a recursive structure for the propagation of the covariance matrix over time. 
Therefore, in order to efficiently approach the persistent monitoring problem, we only need to design a trajectory that will influence the sensing matrices $H_i(k)$ in order to reduce the overall uncertainty. Note that $H_i(k)$ itself is directly related to the covariance matrix through the Kalman filter propagation equations.

In the framework of persistent monitoring where targets have to be visited infinitely often, a very natural assumption is that targets will follow a periodic schedule, as already exploited in some of our previous work on the continuous time version of the problem \cite{pinto2020monitoring,pinto2020multidimensional}. Periodic schedules can approximate arbitrarily well the cost of an optimal schedule \cite{zhao2014optimal} and they naturally provide an upper bound to the inter-visit time of targets. 
One particularly interesting property of periodic schedules is that, under some very natural assumptions, the covariance matrix of each of the targets converges to a steady state periodic solution. Therefore, we only need to design one period of the agent trajectory in order optimize the performance of the system over long horizons.% capture the behavior the covariance matrix as time goes to infinity and, therefore, fully determine the performance of the system over long horizons. %Before proceeding with the discussion about the computation of the steady-state covariance, 

We make the following assumption that ensures that the entire internal state can indeed be estimated from the observations:

\begin{assumption}
    The pair $(A_i,H_{i,\max})$ is observable, for every $i\in\{1,...,M\}$.
\end{assumption}

Under this assumption, it can be shown that, if every target is visited at least once in a period, the covariance matrix will converge to a unique globally attractive steady state solution for any initial conditions $\Sigma_i(0)$, even if the internal state dynamics, captured by the matrix $A_i$, are unstable (the proof is a discrete time version of Prop. 2 in \cite{pinto2020monitoring}, and is omitted here due to space limitations). Therefore, assuming a cyclic trajectory of period $\tau$ where all the targets are visited,
\begin{equation}
    \label{eq:one_period_infinity_cost}
    C=\lim_{M \rightarrow \infty} \frac{1}{M}\sum_{m=1}^M\sum_{i=1}^N\text{tr}(\Sigma_i(m)) = \frac{1}{\tau}\sum_{k=1}^\tau\sum_{i=1}^N\text{tr}(\bar{\Sigma}_i(k))
\end{equation}
where $\bar{\Sigma}_i(k)$ is the steady state covariance matrix and the indexes $i$ and $k$ refer respectively to the index of the target and index of the time step within the cycle. The next section will discuss how to frame the computation of this steady-state covariance matrix as an optimization problem, that later will be used to jointly optimize the cost as expressed in \eqref{eq:one_period_infinity_cost} and the agent trajectory $s(k)$ as given in \eqref{eq:agent_kinematics}.

\section{An optimization approach to computing the infinite horizon cost}
\label{sec:ricatti_eq_opt}
In order to jointly optimize the trajectory and compute the infinite horizon cost using SDPs, we write the Kalman filter in a different format, known as the information filter. We first briefly recall the well known relationship between the Kalman filter and the {information filter. Then we develop a scheme to compute the infinite horizon cost of a periodic trajectory using an SDP whose optimal solution satisfies the information filter equations.}
\subsection{Information Filter}
Recall that the Kalman filter equations can be written in two steps (prediction and update) \cite{thrun2002probabilistic}. The covariance update in the prediction is given by
\begin{equation}
    {\Sigma}_i(k|k-1) = {A_i}{\Sigma}_i(k-1|k-1){A_i}^T+Q_i
\end{equation}
and in the update step by
\begin{multline}
    {\Sigma}_i(k|k) = {\Sigma}_i(k|k-1)-{\Sigma}_i(k|k-1){H}^T_i(k)\\\times(
    {H}_i(k){\Sigma}_i(k|k-1){H}_i^T(k)+{R}_i)^{-1}H_i(k){\Sigma}_i({k|k-1})
\end{multline}
where ${\Sigma}_i(k|k) = \Sigma_i(k)$ and $\Sigma_i(k|k-1)$ are the covariances at time $k$ using information up to time $k$ and $k-1$ respectively. For details, see, e.g. \cite{anderson2012optimal,thrun2002probabilistic}. Before moving to the information filter, we state the following assumption:
\begin{assumption}
    $Q_i$, $R_i$ and the initial covariance matrix $\Sigma_i(0)$ are positive definite, for every $i\in\{1,...,N\}$.
\end{assumption}

Under this assumption, we know that the covariance matrices ${\Sigma}_i(k|k-1)$ and ${\Sigma}_i(k|k)$ are positive definite. Let us define ${P}_i(k|k-1)={\Sigma}_i^{-1}(k|k-1)$ and ${P}_i(k|k)={\Sigma}_i^{-1}(k|k)$. 
Using the matrix inversion lemma in the prediction step \cite{thrun2002probabilistic} yields
\begin{multline}
    P_i(k|k-1) = {Q}_i^{-1}\\-{Q}_i^{-1}A_i(Q_i^{-1}+A^T_iP_i({k-1|k-1})A_i)^{-1}A_i^TQ_i^{-1}.
\end{multline}
Analogously, using the matrix inversion lemma on the update step leads to
\begin{equation}
    P_i({k|k}) = P_i({k|k-1})+H^T_i(k)R_iH_i(k).
\end{equation}
Merging both steps, we get the following recursion:
\begin{multline}
    \label{eq:information_filter_recursion}
    P_i(k|k) = {Q}_i^{-1}+H^T_i(k)R_iH_i(k)\\-{Q}_i^{-1}A_i(Q_i^{-1}+A^T_iP_i({k-1|k-1})A_i)^{-1}A_i^TQ_i^{-1}, 
\end{multline}
which is the well known information filter recursion \cite{anderson2012optimal} that we will use from this point forward in the paper. The information filter is optimal, since it consists simply of a rearrangement of the Kalman filter equations, where instead of propagating directly the covariance, the inverse of the covariance is propagated.

\subsection{Cyclic Schedules and the algebraic Ricatti equation}
Equation \eqref{eq:one_period_infinity_cost} shows that the average steady state uncertainty over a single cycle is equal (in the limit) to the mean squared estimation error over a very long time horizon. Therefore, computing these steady state covariance matrices is an essential part of solving the persistent monitoring problem. In this subsection, we discuss one method for computing $\bar{\Sigma}_i(k)$. We pick this specific method due to the fact that it can be easily integrated in the SDP framework that we will explore in the next subsection. We take an approach similar to \cite{fujimoto2016periodic}, where the steady state covariance is computed using a single augmented algebraic Ricatti equation (ARE). We point out that we cannot directly use the results in \cite{fujimoto2016periodic} because they use the Kalman filter in its standard form and not the information version. %Note that differently from \cite{fujimoto2016periodic}, we work with the information version of the Kalman filter and for this reason we cannot directly apply their results. 
Therefore, in order to use an ARE to compute the steady state behavior of the system, we define the following augmented inverse covariance
\begin{equation}
    \tilde{{P}}_{i,k} = \begin{bmatrix}
    \bar{P}_{i}(k) &\cdots & 0\\
    \vdots & \ddots & \vdots\\
    0 & \cdots & \bar{P}_{i}(k+\tau-1)
    \end{bmatrix},
\end{equation}
{with $\bar{P}_i(k)=\bar{\Sigma}_i^{-1}(k)$, and augmented parameters $\tilde{\Lambda}_{i}=\text{diag}(A_i,...,A_i)$, $\tilde{\Psi}_{i}=\text{diag}(Q_i,...,Q_{i})$, $\tilde{H}_{i,k}=\text{diag}(H_{i}(k),...,H_{i}(k+\tau-1))$, $\tilde{R}_i=\text{diag}(R_i,...,R_i)$. The recursion in \eqref{eq:information_filter_recursion} can be rewritten as:}
\begin{multline}
    \label{eq:augmented_recursion}
    \tilde{P}_{i,k} = \tilde{\Psi}^{-1}_i+\tilde{H}_{i,k}^T\tilde{R}_i\tilde{H}_{i,k}\\-\tilde{\Psi}^{-1}_i\tilde{\Lambda}_i(\tilde{\Psi}^{-1}_i+\tilde{\Lambda}^T_i\tilde{P}_{i,k-1}\tilde{\Lambda}_i)^{-1}\tilde{\Lambda}^T_i\tilde{\Psi}^{-1}_i.
\end{multline}
For ensuring periodicity, we require that $P_i(k+\tau)=P_i(k)$, therefore,
\begin{equation}
     \tilde{P}_{i,k+1}= J\tilde{P}_{i,k}J^T ,\ J=\begin{bmatrix}
    0_{L_i\times (N-1)L_i} & I_{L_i\times L_i}\\
    I_{(N-1)L_i\times(N-1)L_i} & 0_{(N-1)L_i\times L_i}
    \end{bmatrix}.
\end{equation}
Defining $\tilde{Q}^{-1}_i = (J^T\tilde{\Psi}_iJ) ^{-1}$ and $\tilde{A}_i=J^{-T}\tilde{\Lambda}_i$ and substituting into \eqref{eq:augmented_recursion} we get the following algebraic Riccati equation for computing $\tilde{P}_{i,k}$ for each of the targets $i$:
\begin{multline}
    \label{eq:information_filter_algebraic_equation}
    \tilde{Q}^{-1}_i-\tilde{P}_{i,k}+\tilde{H}^T_{i,k}\tilde{R}_i\tilde{H}_{i,k}\\-\tilde{Q}_i^{-1}\tilde{A}_i(\tilde{P}_{i,k}+\tilde{A}_i^T\tilde{Q}_i^{-1}\tilde{A}_i)^{-1}\tilde{A}_i^T\tilde{Q}_i^{-1}=0.    
\end{multline}

\subsection{Solving the ARE as an SDP}
Even though the most efficient methods for solving AREs do not rely on SDPs, in the path for jointly solving the ARE and optimizing the trajectory, we first describe how to cast the solution of the ARE as an SDP. This will allow us to benefit from the efficient solvers available for SDPs and from their convexity properties in order to efficiently approach the persistent monitoring problem. Moving in this direction, we first introduce a relaxed version of \eqref{eq:information_filter_algebraic_equation}, where equality is replaced by inequality, with the goal that, in the optimal solution of the optimization, the constraint will be tight and equality will hold:
\begin{multline}
    \label{eq:information_filter_inequality_algebraic_equation}
    \tilde{Q}^{-1}_i-\Pi_{i}+\tilde{H}^T_{i,k}\tilde{R}_i\tilde{H}_{i,k}\\-\tilde{Q}_i^{-1}\tilde{A}_i(\Pi_{i}+\tilde{A}_i^T\tilde{Q}_i^{-1}\tilde{A}_i)^{-1}\tilde{A}_i^T\tilde{Q}_i^{-1}\succcurlyeq0,
\end{multline}
where $\succcurlyeq0$ denotes that the matrix is positive semi-definite and $\Pi_i$ is a variable for which we want $\Pi_i=\tilde{P}_{i,k}$ as in \eqref{eq:information_filter_algebraic_equation} in the optimal solution of the optimization. {We recall that whenever every target is visited, \eqref{eq:information_filter_algebraic_equation} will have a solution and hence the constraint \ref{eq:information_filter_inequality_algebraic_equation} will be feasible.} Using the Schur complement \cite{balakrishnan1995connections}, this inequality can be written as:
\begin{equation}
    \label{eq:information_filter_inequality}
    \begin{bmatrix}
    \tilde{Q}^{-1}_i-{\Pi}_{i}+\tilde{H}^T_{i,k}\tilde{R}_i\tilde{H}_{i,k} & \tilde{Q}_i^{-1}\tilde{A}_i\\
    \tilde{A}_i^T\tilde{Q}_i^{-1} & {\Pi}_{i}+\tilde{A}_i^T\tilde{Q}_i^{-1}\tilde{A}_i
    \end{bmatrix}\succcurlyeq 0.
\end{equation}
 We also define an upper bound $\Gamma_i$ on the covariance matrix $\Gamma_i \succcurlyeq \Pi_i$, which in the optimal solution will coincide with the covariance matrix. Using Schur's complement, this upper bound can be expressed as:
\begin{equation}
    \label{eq:upper_bound_P_inverse}
    \begin{bmatrix}
    \Gamma_i & I \\ I & {\Pi}_{i}
    \end{bmatrix} \succcurlyeq 0.
\end{equation}

Now, we show that using an SDP, we can compute the exact solution of the information filter and that the relaxations we proposed will indeed be tight in an optimal solution. Moreover, we show that the cost function of the optimization is equal to the trace of the augmented steady state covariance matrix.
Inspired by \cite{balakrishnan1995connections}, where it is shown that the LQR Ricatti equation can be solved as an SDP, we give the following proposition:
\begin{proposition}
\label{prop:exact_solution_information_filter}
If the pair $(\tilde{A}_i,\tilde{H}_{i,k})$ is observable and $Q_i$ and $R_i$ are positive definite, then the optimal solution of the following SDP is such that $\Pi_i^*=\tilde{P}_{i,k}$ is a solution of the ARE \eqref{eq:information_filter_algebraic_equation} and $\Gamma_i^*=(\Pi_i^*)^{-1}$.

\begin{equation}
    \label{eq:optimization_solve_riccati_equation}
  \begin{aligned}
    & \underset{\Gamma_i,\Pi_i}{\text{min}} &&
      \text{tr}(\Gamma_i) \\
    & \text{s.t. }&&\eqref{eq:information_filter_inequality},\qquad  \eqref{eq:upper_bound_P_inverse},\qquad \Gamma_i,\ \Pi_i \succcurlyeq 0.
  \end{aligned}
\end{equation}
\end{proposition}

The proof of Prop. 1 is given in Appendix 1. In the optimal solution of the SDP, $\Gamma_i^*$ is the augmented covariance matrix, i.e. $\Gamma_i^*=\tilde{P}_{i,k}^{-1}$. Therefore, minimizing $\text{tr}(\Gamma_i)$ for all the targets is equivalent to minimizing $\tau^{-1}\sum_{i=1}^{\tau}\text{tr}(\bar{\Sigma}_{i}(k))$, which is the optimization objective in the persistent monitoring problem, as expressed in \eqref{eq:one_period_infinity_cost}. Thus, solving \eqref{eq:optimization_solve_riccati_equation} gives us the squared estimation error of a single target, given an agent trajectory. Although Prop. \ref{prop:exact_solution_information_filter} is not directly used to solve the Persistent Monitoring problem, it gives important insight into Prop. \ref{prop:entire_problem_sdp}, which is the main result of this paper.

\section{Optimization of Persistent Monitoring Schedules}
\label{sec:entire_opt}

In order to obtain an efficient persistent monitoring schedule, one has to design an agent trajectory that will lead to low uncertainty in the estimation. Therefore, in this section, we give a procedure to jointly optimize the steady state uncertainty \eqref{eq:one_period_infinity_cost} and the trajectory of the agent. If we knew in advance when $\norm{s(k)-x_i}$ was larger than $r_i$, then \eqref{eq:sensing_model} would be linear with the squared distance between the agent and target $i$ ($d_i^2(k)$) at every time for every $i$, and the problem would be an SDP. However, since we do not know whether or not $\norm{s(k)-x_i}$ is larger than $r_i$, a set of SDPs needs to be solved in order to obtain the optimal trajectory with that period. Therefore, we propose in this section a two-step procedure, where in the higher level, an algorithm produces sequences of targets to be visited by the agent and determines which of the modes of \eqref{eq:sensing_model} is active in each time step. The lower level, on the other hand, assumes a fixed mode in \eqref{eq:sensing_model} given by the higher level algorithm and through the solution of an SDP produces a trajectory that minimizes the steady state uncertainty.

\subsection{Lower Level Problem}
%First, we describe the lower level algorithm, which is the main result of this paper. 
Recalling \eqref{eq:sensing_model}, in order to simplify notation for the rest of this subsection, we define $G_i=H_{i,\max}^TR_iH_{i,\max}$ and its augmented version $\tilde{G}_i=diag(G_i,...,G_i)$. Note that $H_{i,k}^TR_iH_{i,k}=\gamma_i(k)G_i$ and in the optimization $\gamma_i(k)$ will be treated as a decision variable and $G_i$ as a constant. Therefore,
\begin{equation}
    \tilde{H}_{i,k}^T\tilde{R}_i\tilde{H}_{i,k}=\underbrace{\begin{bmatrix}
    \gamma_i(k)I_{L_i\times L_i} &\cdots & 0\\
    \vdots & \ddots & \vdots\\
    0 & \cdots & \gamma_i(k+\tau)I_{L_i\times L_i}
    \end{bmatrix}}_{\tilde{\gamma}_{i,k}}\tilde{G}.
\end{equation}

Moreover, we create optimization variables $d_i^2$  such that $d_{i}^2(k) \geq \norm{s(k)-x_i}^2$. The underlying goal of creating this variable is that the constraint will be binding in an optimal solution, i.e., $d_{i}^2(k) = \norm{s(k)-x_i}^2$, therefore we can compute $\gamma_{i}(k)$ using \eqref{eq:sensing_model} once $d_i^2(k)$ is fixed. 

Finally, we define logical variables $b_{i,k}\in\{0,1\}$ (that will be fixed pre-defined variables in the optimization) as
\begin{equation}
    b_{i,k}=\begin{cases}
        0,& \norm{s(k)-x_i}>r_i,\\
        1,& \norm{s(k)-x_i}\leq r_i.
    \end{cases}
\end{equation}
These logical variables represent whether or not the agent visits a given target $i$ (i.e., the target within the agent's sensing range) at time step $k$ of the cycle and thus define the mode in \eqref{eq:sensing_model}. With that in mind, we state Prop. \ref{prop:entire_problem_sdp}.

%The informality of the proposition consists in the fact that we do not restate the persistent monitoring problem, already explained in Sec. \ref{sec:formulation}.
\begin{proposition}
    \label{prop:entire_problem_sdp}
    For fixed values of cycle length $\tau$ and logical variables $b_{i,k}$, the solution of the optimization \eqref{eq:entire_optimization}, when it exists, minimizes the cost in \eqref{eq:one_period_infinity_cost}, and the optimal trajectory $s^*(\cdot)$ satisfies dynamic constraints \eqref{eq:agent_kinematics}.
    \begin{mini}
     {\Gamma_i,\Pi_i,s(\cdot)}{\frac{1}{\tau}\sum_{i=1}^N\text{tr}(\Gamma_i)}{}{\label{eq:entire_optimization}}
  \addConstraint{\begin{bmatrix}
    \tilde{Q}^{-1}_i-\Pi_i+\tilde{\gamma}_{i,k}\tilde{G}_i & \tilde{Q}_i^{-1}\tilde{A}_i\\
    \tilde{A}_i^T\tilde{Q}_i^{-1} & \Pi_i+\tilde{A}_i^T\tilde{Q}_i^{-1}\tilde{A}_i
    \end{bmatrix}\succcurlyeq 0}{}%
    \addConstraint{    \begin{bmatrix}
    \Gamma_i & I \\ I & \Pi_i
    \end{bmatrix} \succcurlyeq 0}{}
    \addConstraint{\Gamma_i,\ \Pi_i \succcurlyeq 0}{}
        \addConstraint{\norm{s(k+1)-s(k)}^2\leq u_{\max}^2}{}
    \addConstraint{\norm{s(\tau)-s(1)}^2\leq u_{\max}^2}{}
    \addConstraint{\norm{s(k+1)-x_i}^2\leq d_i^2(k)}{}
    \addConstraint{d_i^2(k)\geq r_i^2,&\ \text{if }b_{i,k}=0}{}
    \addConstraint{\gamma_{i}(k)=0,&\ \text{if }b_{i,k}=0}{}
    \addConstraint{d_i^2(k)\leq r_i^2,&\ \text{if }b_{i,k}=1}{}
    \addConstraint{\gamma_{i}(k)=1-\frac{d_i^2(k)}{r_i^2},&\ \text{if }b_{i,k}=1}{}
    %\addConstraint{ \gamma_i(k) &= \begin{cases}
    %\left(1-\frac{d_j^2(k)}{r_i^2}\right),&\ d_i^2(k) \leq r_i^2,\\
    %0,&\ \text{otherwise}.\end{cases}\;}{}
    \addConstraint{\forall i \in \{1,...,N\},\ \forall k\in\{1,...,\tau\}}.{}
\end{mini}
\end{proposition}
\begin{proof}
We only provide a sketch of the proof due to space limitations and its similarity to Prop. \ref{prop:exact_solution_information_filter}. Using complementary slackness, we conclude that, if the optimal solution of SDP \eqref{eq:entire_optimization} is bounded, then when $b_{i,k}=1$, in the optimal solution $d_i^2(k)=|s(k)-x_i|_2^2$ and therefore $\gamma_i(k)=1-(|s(k)-x_i|_2^2)/r_i^2$. Moreover, using the same arguments as in Prop. \ref{prop:exact_solution_information_filter}, we conclude that $\Pi_i$ in an optimal solution of each feasible subproblem is a solution of the information filter Ricatti equation \eqref{eq:information_filter_algebraic_equation} and that $\Gamma_i=\Pi_i^{-1}$, which means that $\Gamma_i$ is the covariance matrix. Note also that the optimization objective ensures minimization of the infinite horizon cost defined as in \eqref{eq:one_period_infinity_cost} and the constraints ensure that feasible trajectories are periodic and satisfy dynamics constraints \eqref{eq:agent_kinematics}.
\end{proof}
    Some brief insights on \eqref{eq:entire_optimization} are that the first three constraints are used in the solution of the ARE, similar to Prop. \ref{prop:exact_solution_information_filter}. The following two constraints {ensure that the agent movement is feasible according to \eqref{eq:agent_kinematics} and that it is periodic}. The next one is used to compute the distance between the agent and the target and the last four constraints compute $\gamma_i$ based on the relative position of the agent and the target. We also note that constraints that involve the squared norm can be transformed into linear matrix inequalities using the Schur complement, therefore the optimization \eqref{eq:entire_optimization} can be cast as an SDP.

%We note that the optimization problem \eqref{eq:entire_optimization} can be implemented in a way that explores all the different possibilities $b_{i,k}$ for a given $\tau$ using a mixed integer SDP. We refer the interested reader to \cite{bemporad1999control}. %Polyhedral obstacle avoidance can easily be integrated in the framework we proposed, using an approach similar to \cite{richards2002spacecraft}, and is not included on this paper for the sake of conciseness.

\subsection{Higher Level Problem}
Using the optimization problem \eqref{eq:entire_optimization}, we have a procedure such that, for each cycle period $\tau$, we can solve an exponentially growing ($2^{N \times \tau}$) number of SDPs, representing different variations of $b_{i,k}$, and obtain the optimal solution for that cycle period. This approach is very inefficient due to the exponential scalability. We thus propose a graph-based scheme that explores different combinations of $\tau$ and $b_{i,k}$ by exploring how targets are spatially distributed and evaluates each combination using the low level optimization \eqref{eq:entire_optimization}.
%However, in applications of persistent monitoring, exploring the period and the different combinations of $b_{i,k}$ efficiently is essential. Here we propose a simple approach to explore different period lengths while benefiting from the targets spatial structure.

%Note that, in the optimization problem \eqref{eq:entire_optimization}, the problem can be directly converted to an SDP as long as we define whether or not a given target is within the sensing radius of the agent. Therefore, we use the spatial properties of the problem to more efficiently search the space of SDPs, varying including the period. 
%In order to more efficiently search $\tau$ and $b_{i,k}$, we exploit spatial properties of the problem. 
As a motivation for the higher level algorithm we propose, consider the case where two targets are far enough apart that 
the agent seeing one target at a time instant cannot see the other one in the following time step due to the constraints in the dynamics. In a ``blind" exploration of variables $b_{i,k}$, one could encode the possibility of the agent visiting these targets at consecutive time steps and such choice of $b_{i,k}$ renders an infeasible solution of \eqref{eq:entire_optimization}. 
%The agent is not able to observe one of these targets in a given time step and the other one in the following time step due to the distance and therefore a combination of $\tau$ and $b_{i,k}$ that encodes such situation will give an infeasible optimization. 
Moreover, trajectories that are dynamically feasible but have one target that is not visited may lead to unbounded cost and therefore we also do not want to explore them. The goal of the algorithm we introduce in this section is to evaluate only sets of $\tau$ and $b_{i,k}$ that can produce feasible trajectories according to the dynamics and the constraints as in \eqref{eq:agent_kinematics} and also lead to bounded steady state uncertainty. The algorithm we present in this subsection is a ``brute force'' approach and we plan to explore more efficient exploration schemes in future work.%, we avoided a more sophisticated higher level algorithm for the sake of readability.

We abstract the targets as nodes in a graph $\mathcal{G}$. The goal is to find a sequence of nodes to be visited. The cost $\xi(i,h)$ of each edge $(i,h)$ in the graph $\mathcal{G}$ is the minimum number of time steps necessary for the agent to transition between visiting these two targets, i. e.,
%The time the agent spends not visiting nodes does not help the sensing task and therefore we assume that in the graph $\mathcal{G}$,
%Therefore, the edge that links two targets $i$ and $h$ has cost $\xi(i,h)$ equal to
\begin{equation}
    \label{eq:def_edge_cost}
    \xi(i,h)=\max\left(1,\left\lceil \frac{\norm{x_i-x_h}-r_i-r_h}{u_{\max} }\right\rceil\right).
\end{equation}
We point out that visiting a target means being located within its sensing range and not necessarily being exactly at its center. This is the reason why we subtract the radius in the numerator in \eqref{eq:def_edge_cost}. Also, note that if an agent is visiting a given target, it can visit the same target in the following time step. Therefore, the self-transition cost is such that $\xi_{i,i}=1$, for any target. Given this structure, we can directly translate any sequence of visited nodes $\mathcal S=\{n_1,...,n_F\}$ to the number of time steps in a cycle, $\tau$, and to $b_{i,j}$, where
\begin{equation}
\label{eq:compute_length_from_visitation_sequence}
\tau(\mathcal{S}) = \xi(n_F,n_1)+\sum_{j=1}^{F-1}\xi(n_j,n_{j+1}).
\end{equation}
For example, given two nodes in a 2D environment at positions $x_1=(0,0)$ and $x_2=(0,1)$ with $u_{\max}=0.3$ and $r_j=0.3$, then $\xi(1,2)=\xi(2,1)=2$ and $\xi(1,1)=\xi(2,2)=1$. Picking the visiting sequence $\mathcal{S}=\{1,1,2\}$ as an example, its correspondent timestep sequence is $\tau(\mathcal{S})=\xi_{2,1}+\xi_{1,1}+\xi_{1,2}=5$ and $(b_{1,1},...,b_{1,5})=(1,1,0,0,0)$ and $(b_{2,1},...,b_{2,5})=(0,0,0,1,0)$. 

We then propose Algorithm \ref{alg:agents_optimization}, which combines both low and high level stages and we name this {\it SDP-PM}.
The intuition behind this algorithm is that, initially, all the cycles in which each target is visited for exactly one time step are added to a list and ordered according to the number of time steps in that particular cycle. Then, these cycles are explored in order. In the exploration, the cost of that particular cycle is evaluated and all the possibilities of visiting one new (or the same) targets in that cycle are added to the list $\mathcal{L}$.
One thing to note is that the first cycle to be explored will always be the traveling salesman optimal solution and ``simpler" cycles will always be explored first. {We highlight that while Alg. \ref{alg:agents_optimization} provides a feasible solution in every iteration, it is not guaranteed to yield a globally optimal solution with a finite number of iterations ($N_{\text{iter}}$). This algorithm exhaustively searches feasible trajectories that spend the minimum possible time traveling between targets, and this search is heuristically guided in a way where shorter trajectories are explored first. If it was possible to run an infinite amount of iterations, it would theoretically find a globally optimal solution. However, in practical applications the number of iterations is a pre-determined constant $N_{\text{iter}}$ and the designed trajectory is possibly suboptimal.}

\begin{algorithm}
\caption{SDP-PM}
\label{alg:agents_optimization}
\begin{algorithmic}[1]
\State{$OptCost\leftarrow \infty$}
\State{$OptCycle\leftarrow \emptyset$}
\State{$\mathcal{L} \leftarrow \emptyset$}
\For{$\mathcal{S} \in permutation(1,...N)$}
\State Add $(S,\tau(\mathcal{S}))$ to $\mathcal{L}$.
\EndFor
\For{$i \in (1,...,N_{iter})$}
\State{$\mathcal{S}\leftarrow removeFirst(\mathcal{L})$}
\State{$cost=lowerLevelOptimization(\mathcal{S})$}
\If{$cost < OptCost$}
\State{$OptCost\leftarrow cost$}
\State{$OptCycle \leftarrow \mathcal{S}$}
\EndIf
\For{$newVertex\in(1,...,N)$}
\For{$p \in(2,...N_\mathcal{S})$}
\State $\mathcal{S}_{new}\leftarrow \{\mathcal{S}_{1:p-1},newVertex,\mathcal{S}_{p:F_\mathcal{S}}\}$
\State Add $(S_{new},\tau(\mathcal{S}_{new}))$ to $\mathcal{L}$
\EndFor
\EndFor
\EndFor
\State {\bf return} $OptCost,OptCycle$
\end{algorithmic}
\end{algorithm}

\begin{figure}[htp!]
    \centering
    \includegraphics[width=0.12\textwidth]{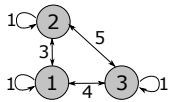}
    \caption{Example graph for illustrating Alg. \ref{alg:agents_optimization}.}
    \label{fig:graph}
\end{figure}

As an illustration of Alg. \ref{alg:agents_optimization}, consider the very simple graph shown in Fig. \ref{fig:graph}. The list is initialized with all the possible cyclic permutations (lines 4-5) of the targets, i.e., $\mathcal{L}=\{(\{1,2,3\},12),(\{1,3,2\},12)\}$, where each element of the list is composed by a sequence of targets and the cycle length in that sequence, computed as in \eqref{eq:compute_length_from_visitation_sequence}. Then, the first node on the list $(\{1,2,3\},12)$ is deleted from the list, its cost is evaluated, and all the cycles that can be constructed from this element by adding one extra visit are added to $\mathcal{L}$. Therefore, when exploring this first element, the cycles added to the list are $(\{1,1,2,3\},13)$, $(\{1,2,2,3\},13)$, $(\{1,3,2,3\},18)$, $(\{1,2,1,3\},14)$, $(\{1,2,2,3\},13)$, $(\{1,2,3,3\},13)$, $(\{1,2,3,1\},13)$, $(\{1,2,3,2\},16)$, and $(\{1,2,3,3\},13)$. However, the list $\mathcal{L}$ is ordered according to cycle length, which means that the next cycle to be explored would be the one with length 12, followed by any cycle with length 13.

%The algorithm we present in this subsection can be improved by simple changes (such as better exploiting the symmetries of the problem, overlapping targets or using a greedy search, rather than an exhaustive one). However, the main contribution of this paper is to present a way to frame the persistent monitoring problem as a convex optimization problem that can be solved exactly, given a fixed sequence of target visits. Searching for this sequence, then, can be interpreted as the problem of finding efficient cycles in a graph. Therefore due to space constraints, we do not discuss extensively search schemes for these cycles of visits. In future works, we plan to address more deeply the impact of specific search schemes. %Note that one particular challenge of this search is that comparing different visitation sequences is very challenging, since we do not have efficient heuristics to guide the search, and in our case we used cycle length.

\section{Simulation Results}
\label{sec:results}

We implemented the SDP-PM Alg. \ref{alg:agents_optimization} with dynamics as in \eqref{eq:model} with parameters
\begin{equation}
    A_i = \textrm{diag}(1.1,1.1),  \quad Q_i = \textrm{diag}(0.1,0.1),
\end{equation}
and observation model as in \eqref{eq:sensing_model} with parameters
\begin{equation}
    H_{\max} = \begin{bmatrix} \frac{1}{\sqrt{2}} & \frac{1}{\sqrt{2}} \\ -\frac{1}{\sqrt{2}} & \frac{1}{\sqrt{2}} \end{bmatrix}, \quad R_i = \textrm{diag}(1,1),\quad r_i=0.6.
\end{equation}
The agent maximum displacement in one time step is bounded by $u_{\max}=0.33.$ The SDP optimization was implemented in MATLAB using YALMIP \cite{Lofberg2004} and solved using MOSEK \cite{mosek}.

To the best of the authors' knowledge, the only approaches proposed in the scientific literature similar enough to be used as a comparison are $RRC$ and its variant $RRC^*$ \cite{lan2016rapidly}. Note that, in its original form, $RRC$ considers a different definition of $\gamma_i(k)$ than Eq. \eqref{eq:sensing_model} and a cost different from \eqref{eq:cost}. However, $RRC$ can be trivially modified to match our definitions. We implemented and compared it to our approach (SDP-PM) in a simple environment and the results are shown in Fig. \ref{fig:results_3_target}. Due to the random nature of RRCs, we ran it 5 independent times and we show its best, worst and average performances.

\begin{figure*}[htp!]
    \centering
    \begin{subfigure}{0.28\textwidth}
        \centering\includegraphics[width=0.8\textwidth]{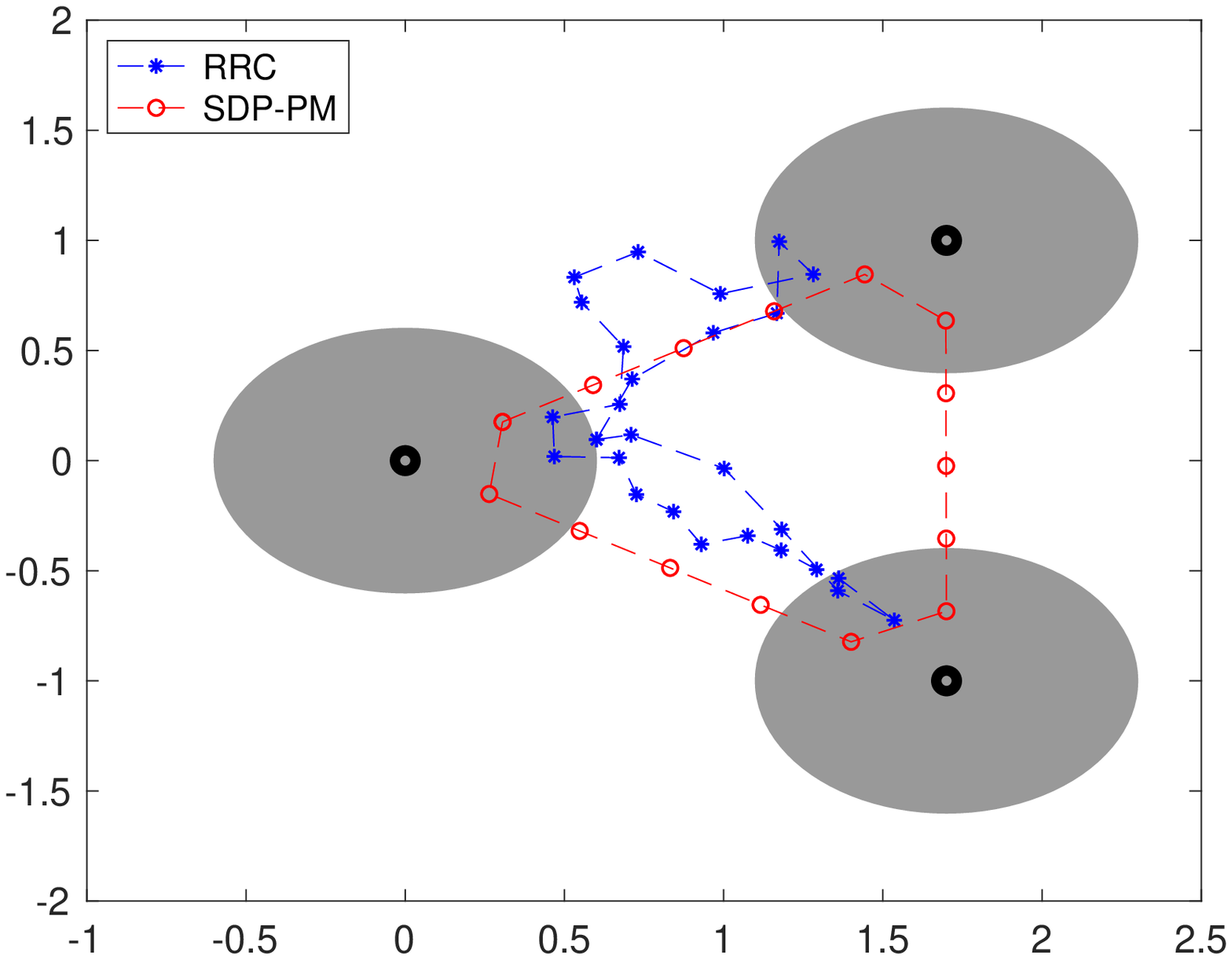}
        \caption{Trajectories using the SDP-PM and RRC algorithms.}
        \label{fig:trajs_comparison}
    \end{subfigure}
    \begin{subfigure}{0.32\textwidth}
        \centering\includegraphics[width=0.8\textwidth]{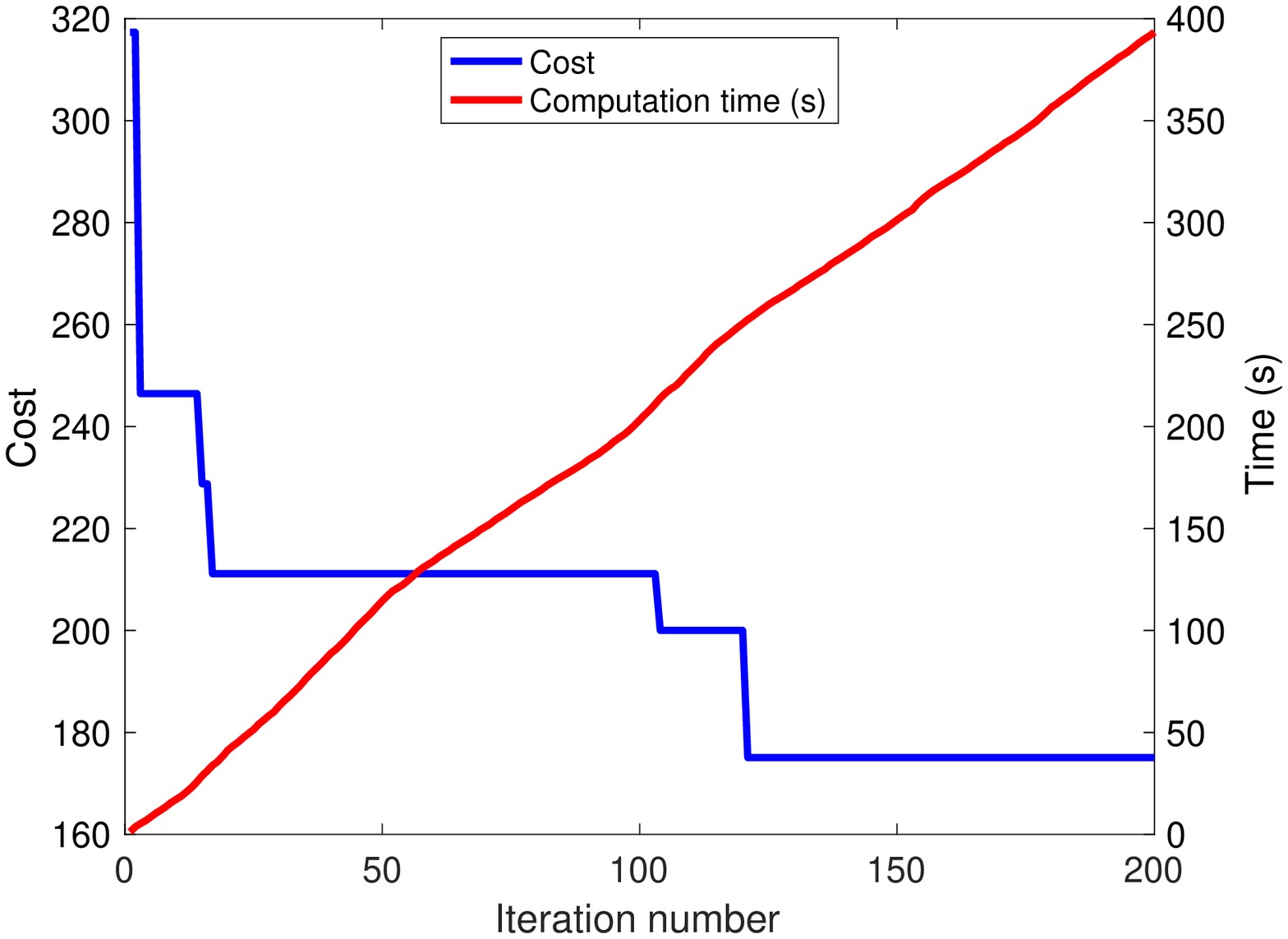}
        \caption{Cost and computation time for SDP-PM.}
        \label{fig:cost_sdp}
    \end{subfigure}
    \begin{subfigure}{0.35\textwidth}
        \centering\includegraphics[width=0.8\textwidth]{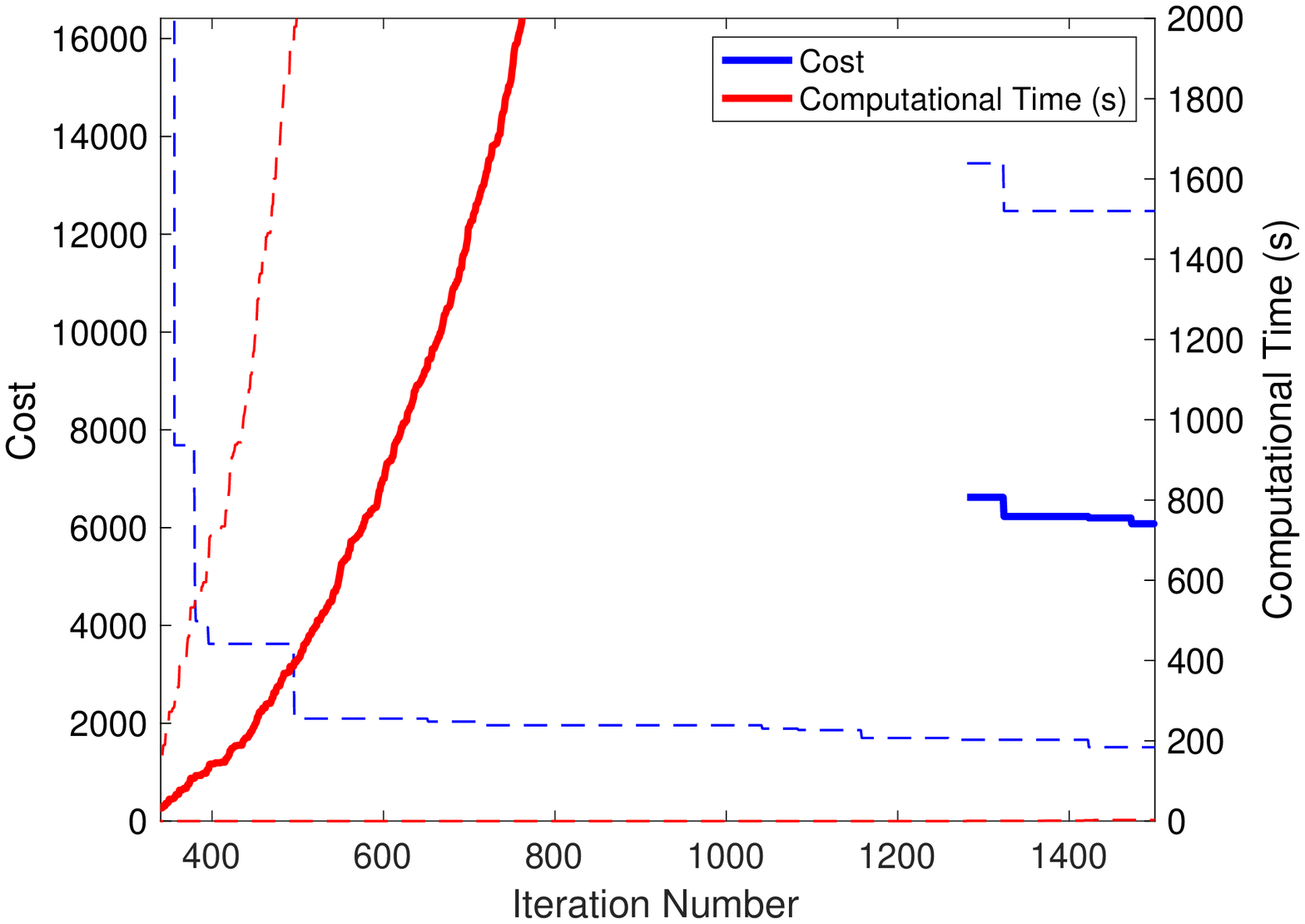}
        \caption{Cost and computation time for RRC.}
        \label{fig:cost_rrc}
    \end{subfigure}
    \caption{Results of the simulation with three targets. (a) Comparison of the trajectories generated by the {\it RRC} and the {\it SDP-PM} approaches. The trajectory displayed for RRC is the one with lowest cost among 5 independent runs of the algorithm. The presented trajectory was obtained after 200 iterations of the SDP-PM algorithm and 1500 iterations of {\it RRC}. The grey area represents the positions for which the agent can sense a given target. (b) Cost and cumulative computation time as a function of the iteration number for {\it SDP-PM}. (c) Cost and cumulative computation time as a function of the iterations of {\it RRC}. The solid lines represent average among 5 runs and the dashed lines are the observed maximum and minimum of the cost and computational time. None of the 5 instances of {\it RRC} found a feasible solution before 345 iterations.}
    \label{fig:results_3_target}
\end{figure*}

In Figs. \ref{fig:cost_sdp} and \ref{fig:cost_rrc}, each iteration of the {\it SDP-PM} consists in exploring one node (lines 6-14 of Alg. \ref{alg:agents_optimization}), while iterations of RRC consist of adding a node to the tree. When considering the computation time of RRC, in order to ensure a fair comparison, we only measured the time spent on solving the ARE, since it is the most computationally demanding part of RRC. The algebraic Ricatti equations were solved using the MATLAB $idare$ function. In the case of the {\it SDP-PM}, the computation time in Fig. \ref{fig:cost_sdp} corresponds to the time spent solving SDPs. Thus, to ensure a fair comparison, we are under reporting the cost of RRC.

From Fig. \ref{fig:trajs_comparison}, one can see that the trajectory produced by {\it SDP-PM} travels between targets in a straight line, while {\it RRC} does not. Also, when the agent visits a target in the {\it SDP-PM} trajectory, it always moves as close as possible to the center of the target (given total time steps and speed limitations), which does not happen in {\it RRC}. The reason is that, for fixed $\tau$ and logical variables $b_{i,k}$, the trajectory generated by {\it SDP-PM} is optimal, while {\it RRC} only has an asymptotic probabilistic notion of optimality, with no deterministic guarantees for a finite number of iterations. Moreover, Figs. \ref{fig:cost_sdp} and \ref{fig:cost_rrc} show that, for reasonable computation times, {\it SDP-PM} produces much better solutions in terms of cost. %Some of the {\it SDP-PM} iterations did take longer than single {\it RRC} iterations, however a much lower number of iterations was necessary to achieve lower costs. In fact, 
The solution found at the first iteration of {\it SDP-PM} has the cost equal to 16.8\% of the cost of the best (in terms of cost) of the 5 runs of {\it RRC} after 1500 iteration. Comparing our approach after 200 iterations and {\it RRC} after 1500, {\it SDP-PM} reports a cost of only 9\% of the best solution of {\it RRC}. We also note that {\it SDP-PM} produced a solution with bounded cost in its first iteration, while {\it RRC} took between 345 and 1305 iterations to find its first feasible solution, i.e. where the target covariances are bounded.

In order to illustrate the performance of our approach in a more complex setup, we ran {\it SDP-PM} in an environment with 7 targets with their centers $x_i$ randomly picked using a uniform distribution in $[0,4]\times[0,4]$. The systems parameters were the same as in the previous case, except that $r_{i,j}$ was set to 0.3. The results are displayed in Figs. \ref{fig:traj_7_targets} and \ref{fig:cost_7_targets}. We note that we ran {\it RRC} 5 times in this environment, with $10^4$ iterations in each trial, and in none of these did {\it RRC} find a feasible solution. By analyzing {\it SDP-PM} results in this more complex environment, we can see that similar to the simpler environment, {\it SDP-PM} finds a feasible solution very fast and refines it within the first few iterations. One interesting aspect of the trajectory generated is that the agent visits some targets for non-consecutive times. This highlights the fact that the approach we propose here does not only locally search around an initial trajectory we input (in this case, initially the TSP solution is the first to be evaluated, where each target is visited once), but also is able to explore trajectories that have major changes in the visiting order compared to the initial exploration schedule. This gives rise to more complex behaviors, such as some targets being explored once, other targets at multiple consecutive time steps, and also targets being visited multiple times but at non consecutive instants. This illustrates the ability of {\it SDP-PM} of introducing a global notion in the search of persistent monitoring schedules.

\begin{figure}[htp!]
    \centering
    \begin{subfigure}{0.22\textwidth}
        \centering \includegraphics[width=\textwidth]{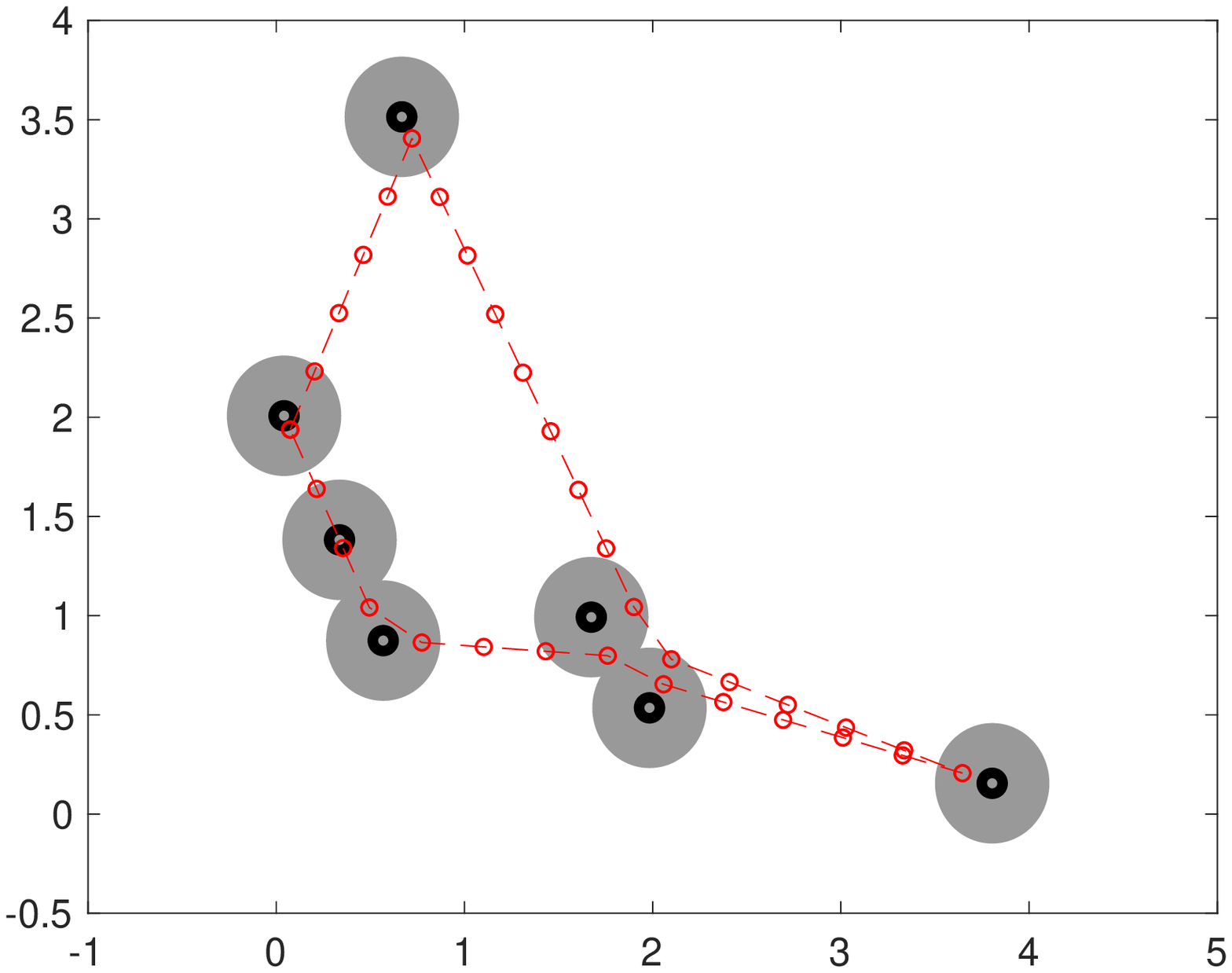}%
        \caption{SDP-PM trajectories.}
        \label{fig:traj_7_targets}
    \end{subfigure}
    \begin{subfigure}{0.25\textwidth}
        \includegraphics[width=\textwidth]{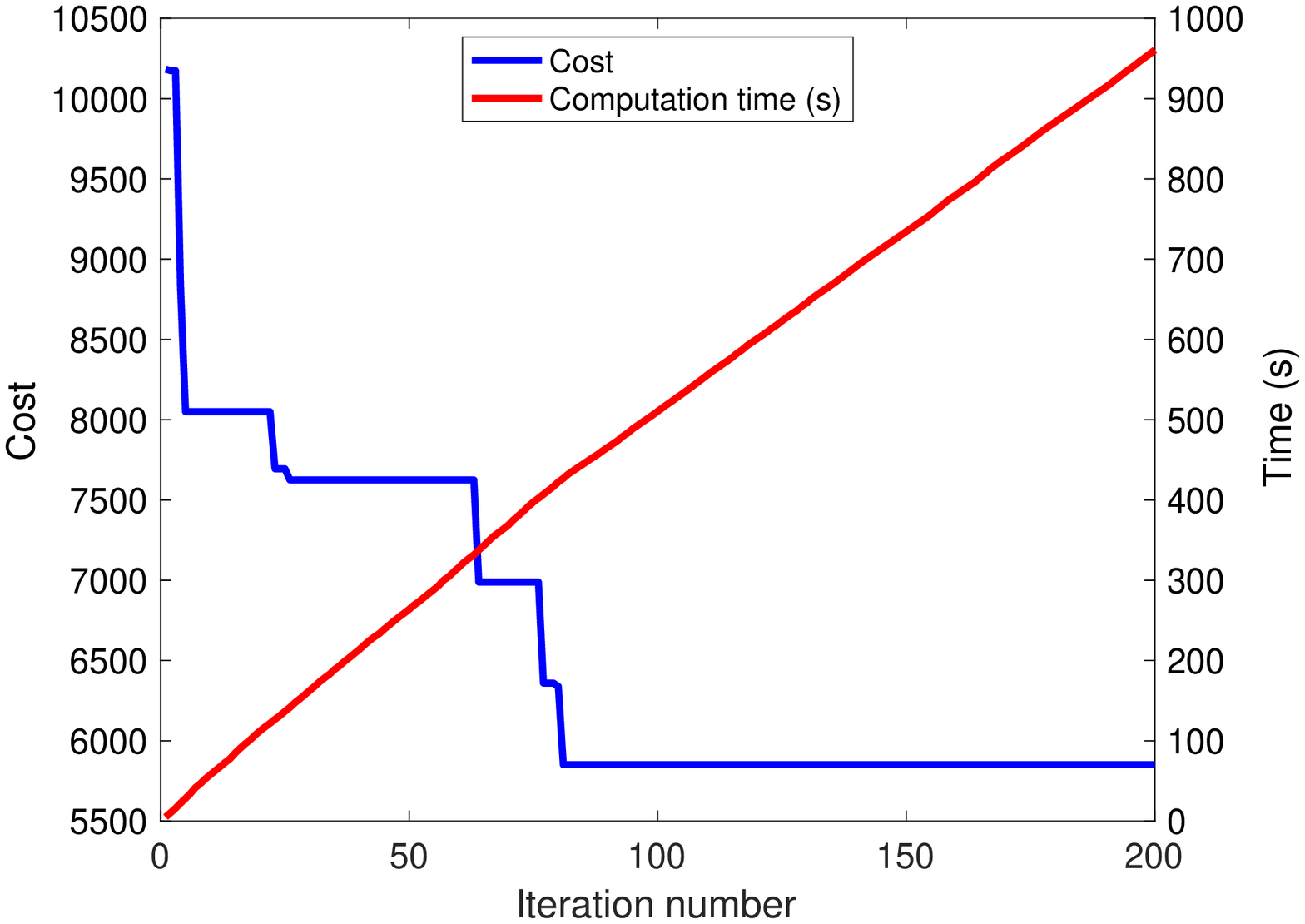}
        \caption{Cost and computation time.}
        \label{fig:cost_7_targets}
    \end{subfigure}
    \caption{The trajectory of the agent after 200 iterations of SDP-PM at more complex environment is given in (a). The circular red marks represent the positions where the agent was at the discrete time steps. The evolution of cost and computation time are shown in (b).}
    
\end{figure}
%\begin{figure}[htp!]
%    \centering
%    \includegraphics[width=0.32\textwidth]{cost_7_targets.eps}
%    \caption{Evolution of optimization cost and computational time using SDP-PM in the environment with 7 agents.}
%    \label{fig:cost_7_targets}
%\end{figure}

%Note that even though it seems that in both scenarios {\it SDP-PM} has converged after 25 iterations, that may not be the case. In Algorithm \ref{alg:agents_optimization}, for each node explored $F_{\mathcal{S}}\tau$ new nodes are added to the set $\mathcal{L}$. Therefore, the number of nodes to be explored grows very fast. However, this algorithm explores trajectories with a lower number of steps (i.e., simpler trajectories) first. As a result, it obtains very good trajectories with some sense of optimality very fast. However, note that {\it SDP-PM} does not scale well for environments with many targets. The solution of lines 3-4 of Alg. \ref{alg:agents_optimization} is equivalent to  solving a traveling salesman problem, which is a NP-Complete problem. 
We note that SDP solvers might present numerical issues as the problem grows, both in terms of number of targets and time steps. From our experience, in trajectories with a number of time steps beyond 100, numerical issues arise, even for a small number of target (4 targets, for instance).  In these cases, a gap between the cost of the primal and dual solutions of the SDP is observed, even after the interior point algorithm converges. RRC does not have these same numerical issues, since it only solves AREs, not SDPs. 

\section{Conclusion and Future Work}
\label{sec:conclusion}

In this paper, we have presented an approach to compute persistent monitoring trajectories for systems that evolve in discrete time, with linear, stochastic dynamics. Given the optimal estimator (the Kalman Filter), we showed that the cost can be jointly optimized with the position of the agents, in a way that minimizes the infinite horizon cost. The simulation results have shown that our approach efficiently solves small to moderately sized problems and that it significantly improves the state of the art, both in terms of computational cost and performance. It is also not restricted to a local optimization, but also explores very diverse trajectories.

For future work, we intend to extend the proposed technique to multi-agent systems and plan to analyze ways to overcome the numerical issues with the SDP that arise when the size of the problem grows. We also want to improve the higher level search algorithm, enhancing the way we explore the candidate optimal cycles, possibly by greedily exploring from an initial candidate schedule. Finally, we plan to investigate the connections between the discrete time persistent monitoring and its continuous time version, explored in \cite{pinto2020monitoring} and possibly combine them, trying to benefit from the global notion of search of {\it SDP-PM} and the trajectory refinement and scalability of our previous work.
\bibliographystyle{IEEEtran}
\bibliography{references.bib}
\appendices
\section{Proof of Proposition \ref{prop:exact_solution_information_filter}}
First we show that Slater's condition \cite{balakrishnan1995connections} is satisfied on the dual problem. The dual of this SDP is:
\begin{mini}
     {Z,Y}{\text{tr}\left(Z\mathcal{M}+Y\begin{bmatrix}
    \Gamma_i & 0 \\ 0 & \Pi_i
    \end{bmatrix}\right)}{}{\label{eq:dual_optimization_solve_riccati_equation}}
  \addConstraint{-Z_{11}+Z_{22}+Y_{22}= 0}{}%
    \addConstraint{ Y_{11} = I }{}
    \addConstraint{Z,\ Y \succcurlyeq 0,}{}
\end{mini}

where
\begin{equation*}
    Z = \begin{bmatrix} Z_{11} & Z_{12}\\Z_{12}^T&Z_{22}\end{bmatrix},\  Y = \begin{bmatrix} Y_{11} & Y_{12}\\Y_{12}^T&Y_{22}\end{bmatrix},
\end{equation*}
\begin{equation*}
    \mathcal{M}=\begin{bmatrix}
    \tilde{Q}^{-1}_i+\tilde{H}^T_{i,k}\tilde{R}_i\tilde{H}_{i,k} & \tilde{Q}_i^{-1}\tilde{A}_i\\
    \tilde{A}_i^T\tilde{Q}_i^{-1} & \tilde{A}_i^T\tilde{Q}_i^{-1}\tilde{A}_i
    \end{bmatrix}.
\end{equation*}
If we pick $Y_{12}=0$, $Y_{11}=0$ and $Z_{11}\succ Z_{22} \succ 0$, then $Z\succ 0$ and $Y\succ 0$ (i.e. strictly positive definite), and $Y$ and $Z$ are feasible, therefore the dual is strictly feasible. On top of that, given that the system is observable and $Q_i$ and $R_i$ are full rank, then Eq. \eqref{eq:information_filter_algebraic_equation} has a unique positive definite solution. Therefore, the primal is feasible since $\Pi_i=\tilde{P}_{i,k}^{-1}$ is a solution of the primal. Thus, strong duality and complementary slackness hold. Now, using complementary slackness, we know that in an optimal solution,
\begin{equation}
    \begin{bmatrix} Y_{11}^* & Y_{12}^*\\(Y_{12}^*)^T&Y_{22}\end{bmatrix}\begin{bmatrix}
    \Gamma_i^* & I \\ I & \Pi_i^*
    \end{bmatrix} = 0,
\end{equation}
which implies that, given that $Y_{11}^*=I$, $\Gamma_i^*=(\Pi_i^*)^{-1} \succ 0$ and that $Y_{22}^*=(\Gamma_i^*)^T\Gamma_i^*$. Therefore, since $Z_{11}=Z_{22}+Y_{22}$, we know that $Z_{11}^*\succ 0$.
Without loss of generality, Z can be expressed as:
\begin{equation*}
    Z=\begin{bmatrix} I \\ K^T\end{bmatrix}Z_{11}\begin{bmatrix} I & K\end{bmatrix}
\end{equation*}
and, given that $Z_{11}^*\succ 0$ complementary slackness also implies that
\begin{equation}
    \begin{bmatrix} I & K^*\end{bmatrix}\begin{bmatrix}
    \tilde{Q}^{-1}_i-\Pi_i^*+\tilde{H}^T_{i,k}\tilde{R}_i\tilde{H}_{i,k} & \tilde{Q}_i^{-1}\tilde{A}_i\\
    \tilde{A}_i^T\tilde{Q}_i^{-1} & \Pi_i^*+\tilde{A}_i^T\tilde{Q}_i^{-1}\tilde{A}_i
    \end{bmatrix} = 0
\end{equation}
which yields 
\begin{multline}
    \label{eq:information_filter_algebraic_equation_2}
    \tilde{Q}^{-1}_i-\Pi_i^*+\tilde{H}^T_{i,k}\tilde{R}_i\tilde{H}_{i,k}\\-\tilde{Q}_i^{-1}\tilde{A}_i(\Pi_{i,k}^*+\tilde{A}_i^T\tilde{Q}_i^{-1}\tilde{A}_i)^{-1}\tilde{A}_i^T\tilde{Q}_i^{-1}=0.
\end{multline}
\end{document}